\documentclass[sigplan, screen, nonacm]{acmart}
\copyrightyear{2020}
\acmYear{2020}
\setcopyright{acmlicensed}\acmConference[LICS '20]{Proceedings of the 35th Annual ACM/IEEE Symposium on Logic in Computer Science (LICS)}{July 8--11, 2020}{Saarbr\"ucken, Germany}
\acmBooktitle{Proceedings of the 35th Annual ACM/IEEE Symposium on Logic in Computer Science (LICS '20), July 8--11, 2020, Saarbr\"ucken, Germany}
\acmPrice{15.00}
\acmDOI{10.1145/3373718.3394771}
\acmISBN{978-1-4503-7104-9/20/07}

\usepackage[utf8]{inputenc}
\usepackage{amsmath, amssymb, amsthm}
\usepackage{mathrsfs}
\usepackage{tikz-cd}
\usepackage{mdframed}
\usepackage[shortlabels]{enumitem}
\setlist{topsep=2pt, itemsep=2pt, partopsep=2pt, parsep=2pt}
\usepackage{calc}
\usepackage[nameinlink]{cleveref}
\usepackage{ebproof}
\usepackage{stmaryrd}
\usepackage[UKenglish]{isodate}
\usepackage{microtype}
\usepackage{bbm}
\usepackage{xfrac}
\usepackage{adjustbox}
\usepackage{csquotes}
\usepackage{nameref}
\usepackage{float}
\usepackage{stfloats}
\input{definitions.tex}
\newcommand{\svc}{\tuple}
\newcommand{\resp}{resp.}
\newcommand{\ar}{\mathsf{ar}}
\newcommand{\bind}[2]{{(#1) #2}}
\newcommand{\denop}[1]{\denote{\oper{#1}}}
\newcommand{\dn}{\denote}

\newcommand{\ruleref}[1]{\hyperref[rule:#1]{\ref*{poly:#1} \eqref{rule:#1}}}

\newcommand{\msigty}{{\Sigma^*_{\sf{ty}}}}
\newcommand{\msigtmsubst}{{\Sigma^\circledast_{\sf{tm}}}}

\newcommand{\opty}[1][]{{O^{#1}_{\sf{ty}}}}
\newcommand{\mopty}{\opty[*]}
\newcommand{\optm}[1][]{{O^{#1}_{\sf{tm}}}}
\newcommand{\moptm}{\optm[*]}
\newcommand{\moptmsubst}{\optm[\circledast]}

\renewcommand{\disc}{\overline}
\renewcommand{\precomp}[1]{(-)#1}

\renewcommand{\sf}{\mathsf}

\newcommand{\freecart}[1]{\cat{Cart}(#1)}

\newcommand{\pbf}{\mathbf}

\newcommand{\textcite}[1]{\citeauthor*{#1}~\cite{#1}}

\makeatletter
\def\unnumfootnote{\xdef\@thefnmark{}\@footnotetext}
\makeatother
\usepackage{silence}
\usepackage{multicol}
\usepackage{wrapfig}

\begin{document}

\title{Algebraic models of simple type theories}
\subtitle{A polynomial approach}

\author{Nathanael Arkor}
\orcid{0000-0002-4092-7930}
\affiliation{
  \department{Department of Computer Science and Technology}
  \institution{University of Cambridge}
}
\email{}

\author{Marcelo Fiore}
\orcid{0000-0001-8558-3492}
\affiliation{
  \department{Department of Computer Science and Technology}
  \institution{University of Cambridge}
}
\email{}

\begin{abstract}
We develop algebraic models of simple type theories, laying out a framework that extends universal algebra to incorporate both algebraic sorting and variable binding. Examples of simple type theories include the unityped and simply-typed $\lambda$-calculi, the computational $\lambda$-calculus, and predicate logic. 

Simple type theories are given models in presheaf categories, with structure specified by algebras of polynomial endofunctors that correspond to natural deduction rules. Initial models, which we construct, abstractly describe the syntax of simple type theories.
Taking substitution structure into consideration, we further provide sound and complete semantics in structured cartesian multicategories.  This development generalises Lambek's correspondence between the \stlc{} and cartesian-closed categories, to arbitrary simple type theories.
\end{abstract}

\maketitle

\unnumfootnote{This is a preprint of \url{https://doi.org/10.1145/3373718.3394771}, published in \emph{Proceedings of the 35th Annual ACM/IEEE Symposium on Logic in Computer Science (LICS '20)}.}

\section{Introduction}

Universal algebra is a framework for describing a class of mathematical structures: precisely those equipped with monosorted algebraic operations satisfying equational laws. Though such structures are prevalent, there are nevertheless many structures of interest in computer science that do not fit into this framework. In particular, notions of type theory, despite being presented in an algebraic style, cannot be expressed as universal algebraic structures. Herein, we follow the tradition of \emph{algebraic type theory} \cite{fiore2008algebraic, fiore2011algebraic} in describing type theories as the extension of universal algebra to a richer setting, \viz{} that of sorting (\ie{} typing) and variable binding.

There are several reasons to be interested in extending universal algebra in this manner. From the perspective of programming language theory, this is a convenient framework for abstract syntax:
the structure of programming
languages, disregarding the superficial details of concrete syntax.
From a categorical perspective, algebraic type theory provides a precise correspondence between syntactic and semantic structure: the rules of a type theory 
give
a conveniently manipulable internal language for reasoning about a categorical structure, which, in turn, models the theory.
The classical result due to Lambek~\cite{lambek1980from}, that the \stlc{} is an internal language for cartesian-closed categories, is a representative example of such a correspondence.

In this paper, we consider the syntax and semantics of \emph{simple type theories}: algebras with sorted binding operations, whose type structure itself is (nonbinding) algebraic. Simple type theories encompass many familiar examples beyond algebraic theories, including the unityped and simply-typed \lci{}, the \clc{} \cite{moggi1988computational}, and predicate logic. Similar extensions to universal algebra have been explored in the past \cite{birkhoff1970heterogeneous, fiore1999abstract, fiore2010second, fiore2013multiversal}, but previous approaches have proven difficult to extend to the dependently-sorted setting that is necessary to describe more sophisticated type theories such as \MLTT{} \cite{martin1984intuitionistic}. We describe a new approach, combining the theories of abstract syntax \cite{fiore1999abstract} and polynomial functors \cite{gambino2013polynomial}, which we feel is an appropriate setting 
to consider dependently-sorted extensions.

\subsection*{Philosophy}

Type theories are typically presented by systems of natural deduction rules describing the inductive structure of the theory. Models of the type theory will therefore have corresponding structure. The observation that motivates our approach may be summarised by the following thesis.

\begin{center}\em
    Natural deduction rules are syntax for polynomials.
\end{center}

\noindent In this paper we give an exposition of this idea, describing a polynomial approach to the semantics of simple type theories. Concretely, we will show how the natural deduction rules presenting the algebraic structure of a simple type theory, which are described precisely by a class of arities, induce polynomial functors in presheaf categories whose algebras are exactly models of the type theory. In particular, the initial algebras are the syntactic models, whose terms are inductively generated from the rules. This provides a correspondence between type theoretic and categorical structure. To build intuition for the general setting, we will frame the classical example of the \stlc{} in this new perspective.

The relationship between polynomials and the algebraic structure of type theories was first proposed by \textcite{fiore2012discrete}, in the context of generalised polynomial functors between presheaf categories. Though our approach is similarly motivated, the setting is different: we consider traditional polynomial functors between slice categories. This is a setting that has been more widely studied \cite{abbott2003categories, gambino2013polynomial} and one we suggest also extends more readily to modelling dependent type theories: \citeauthor{awodey2018polynomial} \cite{awodey2018natural, awodey2018polynomial, newstead2018algebraic}, for instance, have also considered a relationship between polynomial pseudomonads and natural models of type theory~\cite{awodey2018natural}. Their setting, however, is entirely semantic, and one in which the significance of polynomials in the structure of natural deduction rules is not apparent.

In this framework, we consider two classes of models: models of 
\emph{simply typed syntax}~(\Cref{sec:models-syntax}), and models of 
\emph{simple type theories}~(\Cref{sec:models-theories}). Both classes of models have algebraic type structure and multisorted binding (\ie{}~second-order) algebraic term structure, but simple type theories extend syntax in two ways: while \emph{syntax} here refers to those terms solely built inductively from natural deduction rules, \emph{type theories} additionally have an associated notion of (capture-avoiding) substitution: a variable in a term may be replaced by a term of the same type, taking care not to bind any free variables. Typically, a syntax gives rise to a type theory, as one can add a substitution operation that commutes with the operators of the syntax. For this reason, many models of universal algebra do not draw a distinction between syntax and type theory: for instance, Lawvere theories~\cite{lawvere1963functorial} have a built-in notion of substitution, given by composition of morphisms. However, it is useful to consider these two notions separately: substitution gives rise to rich structure that one can only observe by treating it explicitly, for example the substitution lemma~(\Cref{thm:substitution-lemma}) that is ubiquitous in treatments of type theory.

We also consider only type theories (and not syntax) to be equational, as modelling equations involves identifying terms that are syntactically distinct.

\subsection*{Contributions}

The main contributions of this paper are the following.
\begin{enumerate}
    \item A new perspective on natural deduction rules, presenting natural deduction rules for formation, introduction and elimination, as the syntax for polynomials in presheaf categories.
    \item A general definition of models of simply typed syntax and simple type theories.
    \item Initiality theorems, giving a construction of the initial models of simply typed syntax and simple type theories.
    \item A correspondence
    between models of simple type theories and classifying multicategories, generalising the classical Lambek correspondence between the \stlc{} and cartesian-closed categories.
\end{enumerate}

This work provides a basis for our ongoing development of algebraic dependent type theory.

\subsection*{Organisation of the paper}

We build up the definition of a simple type theory in parts, presenting the syntax and semantics in conjunction.

\Cref{sec:types} describes the monosorted nonbinding algebraic structure of types, which is standard from universal algebra.
\Cref{sec:contexts} considers variable contexts and introduces models thereof.
\Cref{sec:terms} is the central contribution of the paper and explains how the multisorted binding algebraic structure on terms may be presented by syntax for polynomials corresponding to natural deduction rules.
\Cref{sec:models-syntax} defines categories of models of simply typed syntax and gives a construction of the initial model (\Cref{thm:initiality-syntax}).
\Cref{sec:substitution} introduces substitution structure on terms and establishes a substitution lemma (\Cref{thm:substitution-lemma}).
\Cref{sec:equations} describes equations on terms, which crucially relies on the substitution structure from the preceding section.
\Cref{sec:models-theories} defines categories of models of simple type theories, which extend syntax by having substitution and equational structure, and 
leads to a construction
of the initial model (\Cref{thm:initiality-theory}).
\Cref{sec:classifying-multicategories} demonstrates how models of simple type theories induce structured cartesian multicategories, establishing a generalised Lambek correspondence (\Cref{thm:equivalence-multicategories} and \Cref{cor:lambek}).

\section{Simple types}
\label{sec:types}

We consider types with monosorted nonbinding algebraic structure \`a la universal algebra. The type constructors of the \stlc{} are examples of such algebraic structure; consider the following formation rules.
\[\begin{prooftree}
\Infer0[$\oper{Unit}$-\textsc{form}]{\oper{Unit}\ \oper{type}}
\end{prooftree}\]
\[\begin{prooftree}
\Hypo{A\ \oper{type}}
\Hypo{B\ \oper{type}}
\Infer2[$\oper{Prod}$-\textsc{form}]{\oper{Prod}(A, B)\ \oper{type}}
\end{prooftree}\]
\[\begin{prooftree}
\Hypo{A\ \oper{type}}
\Hypo{B\ \oper{type}}
\Infer2[$\oper{Fun}$-\textsc{form}]{\oper{Fun}(A, B)\ \oper{type}}
\end{prooftree}\]
These types may be modelled by a set $S$ of sorts with a function expressing the denotations of the type constructors. Base types are described by nullary type constructors, as in universal algebra.
\[1 + S^2 + S^2 \xto{[\denop{Unit}, \denop{Prod}, \denop{Fun}]} S\]
This structure is an algebra for the endofunctor on $\Set$ mapping $S \mapsto S^0 + S^2 + S^2$. This is an example of a \emph{polynomial functor} on $\Set$. Polynomial functors are a categorification of the notion of polynomial functions and similarly represent ``sums of products of variables''. Just as a polynomial function is presented by a list of coefficients, polynomial functors are presented by \emph{polynomials}, which are diagrams of the following shape.
\[I \xfrom{s} A \xto{f} B \xto{t} J\]
Such a polynomial in $\Set$ induces a polynomial functor $\Set/I \to \Set/J$, given by the following, where $B_j = t\inv(j)$ and $A_b = f\inv(b)$.
\[(X_i \mid i \in I) \mapsto (\Sigma_{b \in B_j} \Pi_{a \in A_b} X_{s(a)} \mid j \in J)\]
This is slightly more sophisticated than the traditional sum of products: in particular, we also have a notion of reindexing. Clear introductions to polynomial functors are given in \textcite{weber2011polynomials, gambino2013polynomial}.

Type constructors correspond generally to polynomials in $\Set$. Consider the $\oper{Prod}$ type constructor, for instance. It induces the following very simple polynomial.
\[1 \from {1 + 1}
\to 1 \to 1\]
Each summand in the second component corresponds to a premiss in the formation rule. Here, every morphism is trivial, which is a consequence of types being monosorted. We will see more illustrative examples later. The polynomial induces the polynomial functor $(-) \mapsto (-) \times (-)$, algebras for which are sets $S$ with a function $\denop{Prod}: S^2 \to S$ as intended.

Type operators (\ie{} formation rules) are described generally in terms of arities.

\begin{notation}
Let $M: \Set \to \Set$ be the free monoid endofunctor. For any functor $F: \Set \to \Set$, define $F\kleene \defeq M \comp F$.
\end{notation}

\begin{definition}
\label{def:arity}
We define $\ar_k: \Set \to \Set$, for $k \in \nat$, inductively.
\vspace{-1.6ex}
\begin{align*}
    \ar_0 & \defeq \Id{} \\
    \ar_{k + 1} & \defeq {\ar_k}\kleene \times \ar_k
\end{align*}
We call $\ar_k(S)$ the set of \emph{$S$-sorted $k$\ssth{}-order arities}.
\end{definition}

\begin{notation}
We denote by $A_1, \ldots, A_n \to A$ the $S$-sorted first-order arity 
$\big((A_1, \ldots, A_n), A\big) \in \ar_1(S)$. We identify nullary arities with constants and omit the arrow ($\to$) when $n = 0$.
\end{notation}

\begin{notation}
We denote by $\underline{n}$ the set $\{1, \ldots, n\}$, for $n \in \nat$. In particular, $\underline{0}$ is the empty set.
\end{notation}

First-order arities correspond to the operators of (multisorted) universal algebra \cite{birkhoff1970heterogeneous}, though in this setting we are solely concerned with monosorted operators. Specifically, our type operators are represented by $\{*\}$-sorted first-order arities, where $*$ is the unique kind.

In general, an $n$-ary type operator
\[\oper{O} : \underbrace{*, \ldots, *}_{n \in \Nat} \to *\]
corresponds to a type formation rule of the form
\[\begin{prooftree}
\Hypo{A_1\ \oper{type}}
\Hypo{\cdots}
\Hypo{A_n\ \oper{type}}
\Infer3[$\oper{O}$-\textsc{form}]{\oper{O}(A_1, \ldots, A_n)\ \mathsf{type}}
\end{prooftree}\]
where $A_1, \ldots, A_n$ are type metavariables, universally quantified over all types.

An $n$-ary type operator induces a polynomial in $\Set$
\[1 \from \underline{n} \to 1 \to 1\]
intuitively the following.
\[\scalebox{0.9}{
$\{*\} \from \{A_1 : *\} + \dots + \{A_n : *\} \to \{\oper{O}(A_1, \ldots, A_n) : *\} \to \{*\}$
}\]
An algebra for the induced polynomial functor is given explicitly by a set $S$ and a function $\denop{O}: S^n \to S$. We collect the arities into a single \emph{signature}, which completely describes the inductive structure of the types.

\begin{definition}
A \emph{type operator signature}, denoted $\opty$, is given by a list of $\{*\}$-sorted first-order arities.
\end{definition}

\begin{notation}
To aid readability, we will use the following informal notation throughout. The $\rhd$ symbol separates the type metavariables
from a type, term or equation involving them. For example, the notation \[A, B : * \rhd \oper{Prod}(A, B) : *\] specifies a $\{*\}$-sorted first-order arity $*, * \to *$.
\end{notation}

\begin{example}[Formation rules for the \stlc{}]
Let $I \in \Set$ be a finite set of base types.
    \begin{align*}
        & \rhd \oper{Base}_i : * & (i \in I)  \\
        & \rhd \oper{Unit} : * \\
        A, B : * & \rhd \oper{Prod}(A, B) : * \\
        A, B : * & \rhd \oper{Fun}(A, B) : *
    \end{align*}
\end{example}

A type operator signature induces a polynomial (\resp{} polynomial functor), given by taking the coproduct of the polynomials (\resp{} polynomial functors) induced by its elements.

\begin{notation}
We will denote by $\opty$ both a type operator signature and the polynomial
functor $\opty: \Set \to \Set$ it induces.
\end{notation}

The polynomial functor $\opty$ induces a monad giving the closure of a set of type metavariables under the operators of the signature.

\begin{notation}
Given a type operator signature $\opty$, we denote by $\mopty$ the free $\opty$-algebra monad on $\Set$.
\end{notation}

The Eilenberg--Moore category of the monad $\mopty$ is isomorphic to the category of $\opty$-algebras.

\subsection{Equations on types}
\label{subsec:equations-on-types}

We permit types to be identified by means of equational laws. For any $m \in \nat$, the set $\mopty(\underline{m})$ may be considered syntactically as the set of types parameterised by $m$ type metavariables. Each element of $\underline{m}$ acts as a placeholder, which one can substitute for a concrete type,
by the freeness of $\mopty(\underline{m})$ as in the following. A morphism $\mathbf A$ as below corresponds to a family of sorts $(A_i)_{1 \leq i \leq m} \in S^m$.
%
\begin{equation}
\label{eq:type-equation-substitution}
\begin{tikzcd}
	{\opty(\mopty(\underline{m}))} && {\opty(S)} \\
	{\mopty(\underline{m})} && {S} \\
	& {\underline{m}}
	\arrow["{\pbf A}"', from=3-2, to=2-3]
	\arrow["{\eta}", from=3-2, to=2-1]
	\arrow["{\Psi_{\pbf A}}"', from=2-1, to=2-3, dashed]
	\arrow["{\denop{ty}^*}"', from=1-1, to=2-1]
	\arrow["{\denop{ty}}", from=1-3, to=2-3]
	\arrow["{\opty(\Psi_{\pbf A})}", from=1-1, to=1-3]
\end{tikzcd}
\end{equation}

\begin{definition}
An \emph{$\opty$-type equation} is given by a pair 
$\big( {m \in \nat} , {(L,R) \in \mopty(\underline{m})^2} \big)$, representing an equation between types 
$L \jeq R$ parameterised by $m$ metavariables.
\end{definition}

An $\opty$-type equation induces a term monad identifying the terms in the $(L,R)$ pair \cite{fiore2009construction}, intuitively given by quotienting $\mopty$ by the equation.

\begin{definition}
An \emph{equational type signature}, typically denoted $\sigty$, is given by a type operator signature $\opty$ and a list $\eqty$ of $\opty$-type equations.
\end{definition}

\begin{definition}
Given an equational type signature $\sigty = (\opty, \eqty)$, a \emph{$\sigty$-algebra} is an $\opty$-algebra satisfying the equations of $\eqty$.
\end{definition}

\begin{notation}
Given an equational type signature $\sigty$, we denote by $\msigty$ the associated term monad on $\Set$.
\end{notation}

The term monad associated to an equational type signature $\big(\opty, [\,]\big)$ is the free $\opty$-monad $\mopty$. For any list of \text{$\opty$-type} equations $\eqty$, there is a canonical quotient monad morphism $\mopty \epito \msigty$.

\begin{example}[Unityped $\lambda$-calculus]
In the unityped $\lambda$-calculus there is a single type constant $\oper{D} : *$ and a single type constructor $\oper{Fun} : *, * \to *$, where function types are identified with the base constant: $\rhd\ \oper{D} \jeq \oper{Fun}(\oper{D}, \oper{D})$.
\end{example}


\section{Contexts}
\label{sec:contexts}

Type theories have a notion of \emph{(variable) context}, explicitly quantifying the free variables that may appear in a term. Here, we take the contexts of simple type theories to be cartesian: intuitively, lists of typed variables, admitting exchange, weakening, and contraction. Cartesian context structures model the structure of such contexts.

\begin{definition}
Given an equational type signature $\sigty = (\opty, \eqty)$, a \emph{cartesian $\sigty$-typed context structure} for an algebra $\denop{ty}: \msigty(S) \to S$ consists of
\begin{itemize}
    \item a small category $\sct C$, the \emph{category of contexts}, with a specified terminal object $\epsilon$, the \emph{empty context};
    \item a functor $\svc{-}: \disc S \to \sct C$, embedding sorts as \emph{single-variable contexts}, where $\disc S$ denotes the discrete category on a set $S$;
    \item for all $\Gamma \in \sct C$ and $A \in S$, a specified product $\Gamma \times \svc{A}$, \emph{context extension} of $\Gamma$ by a variable of sort $A$.
\end{itemize}
\end{definition}

\begin{notation}
We write $\freecart{S}$ for the free strict cartesian category on a set $S$, given concretely by the opposite of the comma category 
$(\FinSet \injto \Set) \comma (S : \termcat \to \Set)$,
where $\FinSet$ is the skeleton of the category of finite sets and functions.
\end{notation}

\begin{example}
Every algebraic theory \cite{adamek2010algebraic} (that is, a cartesian category) $\sct C$ is an example of a cartesian $\Id{|\sct C|}$-typed context structure (in fact, one closed under concatenation, rather than just extension).
\end{example}

\begin{definition}
\label{def:context-structure-hom}
A \emph{homomorphism of cartesian $\sigty$-typed context structures} from $(\sct C, S) \to (\sct C', S')$ consists of
\begin{itemize}
    \item a functor $H : \sct C \to \sct C'$;
    \item a $\msigty$-algebra homomorphism $h: S \to S'$,
\end{itemize}
such that the following diagram commutes.
%
\[\begin{tikzcd}
	{\sct C \times \disc S} && {\sct C' \times \disc{S'}} \\
	{\sct C \times \sct C} && {\sct C' \times \sct C'} \\
	{\sct C} && {\sct C'} \\
	& {\termcat}
	\arrow["{H}"', from=3-1, to=3-3]
	\arrow["{\epsilon}", from=4-2, to=3-1]
	\arrow["{\epsilon'}"', from=4-2, to=3-3]
	\arrow["{\times}"', from=2-1, to=3-1]
	\arrow["{\id{} \times \svc{-}}"', from=1-1, to=2-1]
	\arrow["{H \times \disc h}", from=1-1, to=1-3]
	\arrow["{\id{} \times \svc{-}'}", from=1-3, to=2-3]
	\arrow["{\times}", from=2-3, to=3-3]
\end{tikzcd}\]
\end{definition}

Cartesian $\sigty$-typed context structures and their homomorphisms form a category.

\begin{proposition}
\label{prop:initial-cartesian-typed-context-structure}
There is a left-adjoint free functor taking sets $S$ to the free cartesian \text{$\sigty$-typed} context structure on $S$, given by $\freecart{\msigty(S)}$ with $\svc{-}$ the canonical embedding.
\end{proposition}

In particular, the free cartesian $\sigty$-typed context structure on $\emptyset$ is the initial object.

\section{Terms}
\label{sec:terms}

We follow the tradition of abstract syntax, initiated in \textcite{fiore1999abstract}, of representing models of terms as presheaves over categories of contexts. In particular, for a cartesian $\sigty$-typed context structure $\sct C$, we consider presheaves $T : \sct C\op \to \Set$ as sets of terms, indexed by their context. For each context $\Gamma \in \sct C\op$, $T(\Gamma)$ is to be regarded as the set of terms with variables in $\Gamma$; while a morphism ${\rho : \Gamma \to \Gamma'}$ in $\sct C\op$, representing a context renaming, induces a mapping $T(\rho) : T(\Gamma) \to T(\Gamma')$ between terms in different contexts.

\begin{notation}
We use the same symbol for a set $S$ (\resp{} function $h: S \to S'$) and any constant presheaf on $S$ (\resp{} any constant natural transformation on $h$).
\end{notation}

The set of sorts $S$ embeds into $\spsh C$ as a constant presheaf: intuitively a presheaf of types that do not depend on their context. In this light, a natural transformation $\tau : T \to S$ in $\spsh C$ is to be regarded as an assignment of types to terms that respects context renaming. The slice category $\spsh C/S$ is thus an appropriate setting for considering typed terms in context. (Note that we work in the fibred setting, rather than the equivalent indexed setting of \textcite{fiore2002semantic}.)

\begin{definition}
A \emph{typed term structure} for a cartesian \text{$\sigty$-typed} context structure $(\sct C, S)$ is an object of $\spsh C/S$, concretely
\begin{itemize}
    \item a presheaf $T$ in $\spsh C$, the \emph{terms};
    \item a natural transformation $\tau: T \to S$, the \emph{assignment of a type} for each term.
\end{itemize}
The type of any term $t \in T(\Gamma)$ is therefore given by $\tau_\Gamma(t)$ (\cf{}~the view taken in \textcite{fiore2012discretetalk} and \citeauthor{awodey2018natural}'s natural models~\cite{awodey2018natural}).
\end{definition}

\begin{example}
\label{eg:presheaf-of-variables}
The \emph{presheaf of variables} for a cartesian \text{$\sigty$-typed} context structure $(\sct C, S)$ forms a typed term structure $\nu : V \to S$ given by the following, where $\yonem$ denotes the Yoneda embedding.
\begin{align*}
    V & \defeq \coprod_{A \in S} \yonem \svc{A} & \nu(\tuple{A, \rho}) & \defeq A
\end{align*}
\end{example}

The presheaf of variables is so called because, for any context $\Gamma$, the set $V(\Gamma)$ is to be regarded as the variables in $\Gamma$. We note that, for all presheaves $X \in \spsh C$ and $A \in S$, one has $X^{V_A} \equiv X(- \times \svc A)$, illustrating that exponentiation by $V_A$ is the same as context extension \cite{fiore1999abstract} (in turn demonstrating that context extension is polynomial).

\begin{proposition}
For all $n \in \nat$, the morphism  ${\nu^n\!:\!V^n \to S^n}$ is representable.
\end{proposition}

Any presheaf of terms may be restricted to just those with a specified type, by taking pullbacks, as in the following example.

\begin{example}
Given a typed term structure $\tau : T \to S$ and a sort $A \in S$, we denote by $T_A$ the presheaf consisting of terms in $T$ whose type is $A$, given by the fibre:
\[\begin{tikzcd}
	{T_A} & {T} \\
	{1} & {S}
	\arrow["{\tau}", from=1-2, to=2-2] \arrow["{\iota_A}", from=1-1, to=1-2] \arrow[from=1-1, to=2-1] \arrow["{A}"', from=2-1, to=2-2] \arrow["\lrcorner"{very near start, rotate=0}, from=1-1, to=2-2, phantom]
\end{tikzcd}\]
It induces a typed term structure
given by the composite $T_A\to 1\xto{A} S$.
\end{example}

\subsection{Algebraic models of the \stlc{}}
\label{sec:algebraic-models-of-the-stlc}

Terms have two additional forms of structure that is not found in simple types: multisorting and binding. We walk through the illustrative algebraic term structure of the \stlc{} to give intuition before providing the general construction in \Cref{sec:algebraic-term-structure}. First, we will identify the structure we expect our models to have, before seeing how this structure arises from our models being algebras for a polynomial functor in \Cref{sec:stlc-polynomials}.

As with the algebraic structure for types, the algebraic structure for terms is presented by natural deduction rules (typically introduction or elimination rules), each rule corresponding to an operator on terms.

\paragraph{Products}

The introduction rule for $\oper{Prod}$ is given by the following.
\begin{equation}
    \label{rule:prod-intro}
    \begin{prooftree}
    \Hypo{\tr{a}{A}}
    \Hypo{\tr{b}{B}}
    \Infer2[$\oper{Prod}$-\textsc{intro}]{\tr{\oper{pair}(a, b)}{\oper{Prod}(A, B)}}
    \end{prooftree}
\end{equation}
Conceptually, the introduction rule allows one to take two terms of any two types $A$ and $B$ and form a new term, their pair, such that the type of the new term is the product $\denop{Prod}(A, B)$, given by the algebraic structure of the types. A typed term structure $\tau : T \to S$ therefore models \ref{poly:prod-intro} when equipped with a morphism $\denop{pair}$ such that the following diagram commutes.
%
\[\begin{tikzcd}
	{T \times T} & {T} \\
	{S \times S} & {S}
	\arrow["{\tau}", from=1-2, to=2-2]
	\arrow["{\tau \times \tau}"', from=1-1, to=2-1]
	\arrow["{\denote{\oper{pair}}}", from=1-1, to=1-2, dashed]
	\arrow["{\denote{\oper{Prod}}}"', from=2-1, to=2-2]
\end{tikzcd}\]
The elimination rules for products are given by the following.
\vspace{-1.6ex}
\begin{equation}
    \label{rule:prod-elim1}
    \begin{prooftree}
    \Hypo{\tr{p}{\oper{Prod}(A, B)}}
    \Infer1[$\oper{Prod}$-\textsc{elim}$_1$]{\tr{\oper{proj}_1(p)}{A}}
    \end{prooftree}
\end{equation}
\begin{equation}
    \label{rule:prod-elim2}
    \begin{prooftree}
    \Hypo{\tr{p}{\oper{Prod}(A, B)}}
    \Infer1[$\oper{Prod}$-\textsc{elim}$_2$]{\tr{\oper{proj}_2(p)}{B}}
    \end{prooftree}
\end{equation}
A typed term structure $\tau : T \to S$ models the first projection when equipped with a morphism $\denop{proj_1}$ such that the following left-hand square commutes, where $T_{\denop{Prod}}$ is given by the following right-hand square.
%
\[\begin{tikzcd}
	{T_{\denote{\oper{Prod}}}} & {T} \\
	{S \times S} & {S}
	\arrow["{\tau}", from=1-2, to=2-2]
	\arrow[from=1-1, to=2-1]
	\arrow["{\denote{\oper{proj_1}}}", from=1-1, to=1-2, dashed]
	\arrow["{\pi_1}"', from=2-1, to=2-2]
\end{tikzcd}
\qquad
\begin{tikzcd}
	{T_{\denote{\oper{Prod}}}} & {T} \\
	{S \times S} & {S}
	\arrow["{\tau}", from=1-2, to=2-2] \arrow["{\iota}", from=1-1, to=1-2] \arrow[from=1-1, to=2-1] \arrow["{\denote{\oper{Prod}}}"', from=2-1, to=2-2] \arrow["\lrcorner"{very near start, rotate=0}, from=1-1, to=2-2, phantom]
\end{tikzcd}
\]
This condition is analogous to the one for the introduction rule, the primary difference being that it is only possible to project from terms that have type $\oper{Prod}(A, B)$ for some types $A$ and $B$. The situation for \ref{poly:prod-elim2} is analogous.

\paragraph{Units}

Compared to that for products, the algebraic structure for units is almost trivial. The introduction rule for $\oper{Unit}$ is given by the following.
\begin{equation}
    \label{rule:unit-intro}
    \begin{prooftree}
    \Infer0[$\oper{Unit}$-\textsc{intro}]{\tr{\oper{u}}{\oper{Unit}}}
    \end{prooftree}
\end{equation}
A typed term structure $\tau : T \to S$ for the $\oper{Unit}$ type should therefore single out a
term $\denop{u}$ with type $\tau(\denop{u}) = \denop{\oper{Unit}}$ (in any context). That is, we expect the following diagram to commute.
\[\begin{tikzcd}
	{1} & {T} \\
	& {S}
	\arrow["{\tau}", from=1-2, to=2-2]
	\arrow["{\denop{Unit}}"', from=1-1, to=2-2]
	\arrow["{\denop{u}}", from=1-1, to=1-2, dashed]
\end{tikzcd}\]

\paragraph{$\lambda$-abstraction}

Having considered sorting structure, we now consider variable binding. The introduction rule for $\oper{Fun}$ is given by the following.
\begin{equation}
    \label{rule:fun-intro}
    \begin{prooftree}
    \Hypo{\trec{\twt{a}{A}}{b}{B}}
    \Infer1[$\oper{Fun}$-\textsc{intro}]{\tr{\oper{abs}(\bind{\twt{a}{A}}{b})}{\oper{Fun}(A, B)}}
    \end{prooftree}
\end{equation}
The $\oper{abs}$ operator allows one to take a term in an extended context and form a term in the original context. A typed term structure $\tau : T \to S$ therefore models \ref{poly:fun-intro} when equipped, for every context $\Gamma$ and types $A, B \in S$, with a 
mapping
$\denop{abs}_{\Gamma}^{A, B}$, natural in $\Gamma$, such that the following diagram on the left commutes.
%
\[\begin{tikzcd}
	{T_B(\Gamma \times \svc{A})} & {T(\Gamma)} \\
	{1} & {S}
	\arrow["{\tau_\Gamma}", from=1-2, to=2-2]
	\arrow["{\denop{abs}_{\Gamma}^{A, B}}", from=1-1, to=1-2, dashed]
	\arrow[from=1-1, to=2-1]
	\arrow["{\denop{Fun}(A,B)}"', from=2-1, to=2-2]
\end{tikzcd}
\qquad
\begin{tikzcd}
	{{T_B}^{V_A}} & {T} \\
	{1} & {S}
	\arrow["{\tau}", from=1-2, to=2-2]
	\arrow["{\denop{abs}^{A, B}}", from=1-1, to=1-2, dashed]
	\arrow[from=1-1, to=2-1]
	\arrow["{\denop{Fun}(A,B)}"', from=2-1, to=2-2]
\end{tikzcd}\]
Through the relationship between context extension and exponentiation by 
representables (\Cref{eg:presheaf-of-variables}),
this is equivalent to the above diagram on the right, for all $A, B \in S$.
Finally, quantifying over the types, this is further equivalent to the following formulation.
\vspace{-1.6ex}
%
\[\begin{tikzcd}
	{\coprod_{A, B \in S} {T_B}^{V_A}} & {T} \\
	{S \times S} & {S}
	\arrow["{\tau}", from=1-2, to=2-2]
	\arrow["{\denop{abs}}", from=1-1, to=1-2, dashed]
	\arrow["{\pi}"', from=1-1, to=2-1]
	\arrow["{\denop{Fun}}"', from=2-1, to=2-2]
\end{tikzcd}\]
The elimination rule is simpler.
\begin{equation}
    \label{rule:fun-elim}
    \begin{prooftree}
    \Hypo{\tr{f}{\oper{Fun}(A, B)}}
    \Hypo{\tr{a}{A}}
    \Infer2[$\oper{Fun}$-\textsc{elim}]{\tr{\oper{app}(f, a)}{B}}
    \end{prooftree}
\end{equation}
A typed term structure $\tau : T \to S$ models function application when equipped with a morphism $\denop{app}$ such that the following square commutes.
%
\[\begin{tikzcd}
	{\coprod_{A, B \in S} T_{\denop{Fun}(A, B)} \times T_A} & {T} \\
	{S \times S} & {S}
	\arrow["{\pi_2}"', from=2-1, to=2-2]
	\arrow["{\tau}", from=1-2, to=2-2]
	\arrow["{\denop{app}}", from=1-1, to=1-2, dashed]
	\arrow["{\pi}"', from=1-1, to=2-1]
\end{tikzcd}\]

\subsection{Polynomials for the \stlc{}}
\label{sec:stlc-polynomials}

We now show how the above structure for models of the \stlc{} is actually algebraic structure for a polynomial functor. While, so far, we have only dealt with polynomials in $\Set$, we recall that the concept makes sense for any presheaf category $\spsh C$.

For every morphism $f : A \to B$ in $\spsh C$, there is an adjoint triple $\dsum{f} \adj \pb{f} \adj \dprod{f}$ where $\dsum{f}$ is postcomposition by $f$ and $\pb{f} : \spsh C/B \to \spsh C/A$ is pullback along $f$. Every polynomial $I \from A \to B \to J$ induces a polynomial functor $\dsum{t} \dprod{f} \pb{s} : \spsh C/I \to \spsh C/J$.

An algebra for a functor $F : \spsh C/S \to \spsh C/S$ is a typed term structure $\tau : T \to S$ with a morphism $\varphi : \text{dom}(F(\tau)) \to T$ such that the following diagram commutes.
%
\[\begin{tikzcd}
	{\text{dom}(F(\tau))} & {T} \\
	& {S}
	\arrow["{\tau}", from=1-2, to=2-2]
	\arrow["{F(\tau)}"', from=1-1, to=2-2]
	\arrow["{\varphi}", from=1-1, to=1-2, dashed]
\end{tikzcd}\]
We will write $F(T)$ to mean $\text{dom}(F(\tau))$ when unambiguous.

In particular, algebras for a polynomial functor, 
$\varphi: \dsum{t} \dprod{f} \pb{s}(\tau : T \to S) \to (\tau : T \to S)$, 
are illustrated by the following diagram.
%
\[\begin{tikzcd}
	{T} & {\pb{s}(T)} & {\dprod{f} \pb{s}(T)} & {T} \\
	{S} & {A} & {B} & {S}
	\arrow["{f}"', from=2-2, to=2-3]
	\arrow["{s}", from=2-2, to=2-1]
	\arrow["{t}"', from=2-3, to=2-4]
	\arrow["{\tau}"', from=1-1, to=2-1]
	\arrow["{\tau}", from=1-4, to=2-4]
	\arrow["{\dprod{f} \pb{s}(\tau)}" description, from=1-3, to=2-3]
	\arrow["{\varphi}", from=1-3, to=1-4, dashed]
	\arrow[from=1-2, to=1-1]
	\arrow["{\pb{s}(\tau)}" description, from=1-2, to=2-2]
	\arrow["\lrcorner"{very near start, rotate=-90}, from=1-2, to=2-1, phantom]
\end{tikzcd}\]
We will sometimes depict polynomials geometrically as in the following.
\vspace{-1.6ex}
%
\[\begin{tikzcd}[cramped, row sep=tiny, column sep=small]
	& {A} & {B} \\
	{I} &&& {J}
	\arrow["{s}"', from=1-2, to=2-1] \arrow["{t}", from=1-3, to=2-4] \arrow["{f}", from=1-2, to=1-3]
\end{tikzcd}\]
We may then unambiguously omit a component morphism, which is taken to be the identity. The composition of two polynomials is also a polynomial \cite[Proposition~1.12]{gambino2013polynomial}, depicted graphically as in the following.
\[\begin{tikzcd}[cramped, row sep=tiny, column sep=small]
	& {A} & {B} && {C} & {D} \\
	{I} &&& {J} &&& {K}
	\arrow[
	  from=1-2, to=2-1] 
	\arrow[
	  from=1-3, to=2-4] 
	\arrow[
	  from=1-2, to=1-3] 
	\arrow[
	  from=1-5, to=2-4] 
	\arrow[
	  from=1-5, to=1-6] 
	\arrow[
	  from=1-6, to=2-7]
\end{tikzcd}\]

\paragraph{Products}

The condition for \ruleref{prod-intro} exactly states that $\tau : T \to S$ is an algebra for the polynomial functor induced by the following polynomial in $\spsh C$.
\begin{equation}
    \label{poly:prod-intro}
    \tag{\hyperref[rule:prod-intro]{$\oper{Prod}$-\textsc{intro}}}
    S \xfrom{[\pi_1, \pi_2]} S^2 + S^2 \xto{\codiag_2} S^2 \xto{\denote{\oper{Prod}}} S
\end{equation}
The structure of this polynomial may seem opaque at first; we will attempt to provide some intuition. The polynomial describes a (many-in, one-out) transformation between terms, respecting the type structure. The middle component $\codiag_2 : S^2 + S^2 \to S^2$ represents the type metavariables $A$ and $B$: each summand in the domain represents the metavariables in a premiss, while the codomain represents the metavariables in the conclusion. While some metavariables may not appear in every premiss, each premiss is implicitly parameterised by each type metavariable; the codiagonal ensures that the metavariables available to each premiss (and the conclusion) are the same (\ie{} unified).

The leftmost component $S \from S^2 + S^2 : [\pi_1, \pi_2]$ describes the types of each premiss, given the metavariables. In this case, the types are simply projections: the left-hand side to $A$ and the right-hand side to $B$. The rightmost component $\denop{Prod} : S^2 \to S$ describes the type of the conclusion, given the metavariables: in this case, constructing the product of $A$ and $B$.

An algebra for the functor induced by this polynomial is calculated explicitly below, to demonstrate that it aligns with the structure we deduced earlier.
%
\[\begin{tikzcd}
	{T} & {T \times S + S \times T} & {T^2} & {T} \\
	{S} & {S^2 + S^2} & {S^2} & {S}
	\arrow["{\denop{Prod}}"', from=2-3, to=2-4] \arrow["{[\pi_1, \pi_2]}", from=2-2, to=2-1] \arrow["{\codiag_2}"', from=2-2, to=2-3] \arrow["{\tau}"', from=1-1, to=2-1] \arrow["{[\pi_1, \pi_2]}"', from=1-2, to=1-1] \arrow["{\tau \times \id{} + \id{} \times \tau}" description, from=1-2, to=2-2] \arrow["{\tau^2 = \dprod{\codiag_2}(\tau \times \id{} + \id{} \times \tau)}" description, from=1-3, to=2-3] \arrow["{\tau}", from=1-4, to=2-4] \arrow["{\denop{pair}}", from=1-3, to=1-4, dashed] \arrow["\lrcorner"{very near start, rotate=-90}, from=1-2, to=2-1, phantom]
\end{tikzcd}\]
The polynomials for the projections are similarly described. For $\tau : T \to S$ to be a model of the product eliminators \ruleref{prod-elim1} and \ruleref{prod-elim2}, we require it to be an algebra for the following polynomials.
\begin{equation}
    \label{poly:prod-elim1}
    \tag{\hyperref[rule:prod-elim1]{$\oper{Prod}$-\textsc{elim}$_1$}}
    S \xfrom{\denote{\oper{Prod}}} S^2 \xto{\codiag_1} S^2 \xto{\pi_1} S
\end{equation}
\begin{equation}
    \label{poly:prod-elim2}
    \tag{\hyperref[rule:prod-elim2]{$\oper{Prod}$-\textsc{elim}$_2$}}
    S \xfrom{\denote{\oper{Prod}}} S^2 \xto{\codiag_1} S^2 \xto{\pi_2} S
\end{equation}
Given some examination, the structure of the polynomials is analogous to that of the introduction rule, with the first component selecting the type of the premiss, the codiagonal (in this case trivially) unifying the premisses, and the final component selecting the type of the conclusion.

\paragraph{Units}

For \ruleref{unit-intro}, the polynomial inducing the structure is similarly defined. For $\tau : T \to S$ to be a model of $\oper{u}$, we require it to be an algebra for the following polynomial.
\begin{equation}
    \label{poly:unit-intro}
    \tag{\hyperref[rule:unit-intro]{$\oper{Unit}$-\textsc{intro}}}
    S \xfrom{!} 0 \xto{\codiag_0} 1 \xto{\denop{Unit}} S
\end{equation}
One may see that, as the introduction rule for $\oper{Unit}$ has no premisses and no type metavariables, this (trivially) fits the same pattern as with $\oper{Prod}$.

\paragraph{$\lambda$-abstraction}

To describe binding structure, we need more sophisticated polynomials. For $\tau : T \to S$ to be a model of \ruleref{fun-intro}, we require it to be an algebra for the following polynomial.
\begin{equation}
    \label{poly:fun-intro}
    \tag{\hyperref[rule:fun-intro]{$\oper{Fun}$-\textsc{intro}}}
    S \xfrom{\pi_2} V \times S \xto{\nu \times \id{}} S^2 \xto{\denote{\oper{Fun}}} S
\end{equation}
Here, the first and last components are familiar from the previous examples. 
The form of the middle component is new: metavariables involved in context extension, and therefore in variable binding, must be fibred over the presheaf of variables $V$. The typed term structure $\nu: V \to S$ of~\Cref{eg:presheaf-of-variables} forgets
the information associated to a variable apart from its type.

For $\tau : T \to S$ to be a model of \ruleref{fun-elim}, we require it to be an algebra for the following polynomial.
\begin{equation}
    \label{poly:fun-elim}
    \tag{\hyperref[rule:fun-elim]{$\oper{Fun}$-\textsc{elim}}}
    S \xfrom{[\denote{\oper{Fun}}, \pi_1]} S^2 + S^2 \xto{\codiag_2} S^2 \xto{\pi_2} S
\end{equation}

Polynomials are closed under taking coproducts: for a typed term structure to be a model of the entire structure of the \stlc{}, therefore, we require it to be an algebra for the polynomial endofunctor induced by the coproduct of all the aforementioned polynomials.

\subsection{Algebraic term structure}
\label{sec:algebraic-term-structure}

We now give syntax for a general natural deduction rule for a term operator and the construction of the polynomial it induces. As with type operators, we have a notion of arity corresponding to term operators.

\begin{notation}
\label{not:second-order-arity}
We denote by \[(A^1_1, \ldots, A^1_{k_1})A_1, \ldots, (A^n_1, \ldots, A^n_{k_n})A_n \to (B_1, \ldots, B_k)B\] the $S$-sorted second-order arity (\Cref{def:arity})
\[\scriptstyle ((((A^1_1, \ldots, A^1_{k_1}), A_1), \ldots, ((A^n_1, \ldots, A^n_{k_n}), A_n)), ((B_1, \ldots, B_k), B)) \in \ar_2(S)\]
Second-order arities correspond to the operators of multisorted binding algebra \cite{fiore2010second}: such an arity represents an operator taking $n$ arguments, the $i$\ssth{} of which binds $k_i$ variables, which is parameterised by $k$ variables. We identify nullary arities with constants.
\end{notation}

Given an equational type signature $\sigty$ and $m \in \Nat$ type metavariables, we can represent term operators by \text{$\msigty(\underline{m})$-sorted} second-order arities. An $n$-ary term operator
\begin{equation}
    \label{eq:second-order-oper}
    \footnotesize \oper{o} : (A^1_1, \ldots, A^1_{k_1})A_1, \ldots, (A^n_1, \ldots, A^n_{k_n})A_n \to (B_1, \ldots, B_k)B
\end{equation}
corresponds to a rule as in \Cref{fig:term-operator}, universally quantified over all contexts $\Gamma$.

\begin{figure*}
\begin{mdframed}
\[\begin{prooftree}
\Hypo{\trec{\twt{x^1_1}{A^1_1}, \ldots, \twt{x^1_{k_1}}{A^1_{k_1}}}{t_1}{A_1}}
\Hypo{\cdots}
\Hypo{\trec{\twt{x^n_1}{A^n_1}, \ldots, \twt{x^n_{k_n}}{A^n_{k_n}}}{t_n}{A_n}}
\Infer3{\trec{\twt{y_1}{B_1}, \ldots, \twt{y_k}{B_k}}{\oper{o}[\twt{y_1}{B_1}, \ldots, \twt{y_k}{B_k}]\big((\twt{x^1_1}{A^1_1}, \ldots, \twt{x^1_{k_1}}{A^1_{k_1}}) t_1, \ldots, (\twt{x^n_1}{A^n_1}, \ldots, \twt{x^n_{k_n}}{A^n_{k_n}}) t_n\big)}{B}}
\end{prooftree}\]
\caption{Natural deduction rule for a term operator}
\label{fig:term-operator}
\end{mdframed}
\end{figure*}

\begin{definition}
\label{def:parameterised-term-operators}
We say that a term operator is \emph{parameterised} when $k \neq 0$.
\end{definition}

A term operator for an equational type signature $\sigty$, as in \eqref{eq:second-order-oper}, induces a polynomial in $\spsh C$ for any cartesian $\sigty$-typed context structure, given in \Cref{fig:term-polynomial}.

\begin{figure*}
\begin{mdframed}
    \[\!\!\!\begin{tikzcd}[column sep=scriptsize, row sep=normal]
    	& & & {\coprod_{1 \leq i \leq n} S^m} & {S^m} && {V^k \times S^m} 
    	\\
    	&
    	{\coprod_{1 \leq i \leq n} V^{k_i} \times S^m} 
    	\arrow{r}[swap, yshift=-1ex]{\coprod_{1 \leq i \leq n} \nu^{k_i} \times \id{}} 
    	& 
    	{\coprod_{1 \leq i \leq n} S^{k_i} \times S^m} 
    	&
    	&& 
    	{S^k \times S^m} 
    	&& 
    	{S} 
    	\\
    	{S} 
    	\arrow["{[\denote{A_i} \comp \pi_2]_{1 \leq i \leq n}}"', from=2-2, to=3-1]
    	\arrow["{\coprod_{1 \leq i \leq n} \tuple{\tuple{\denote{A^i_j}}_{1 \leq j \leq k_i},           \id{}}}"', from=1-4, to=2-3]
    	\arrow["{\nabla_n}", from=1-4, to=1-5]
    	\arrow["{\tuple{\tuple{\denote{B_j}}_{1 \leq j \leq k}, \id{}}}"', from=1-5, to=2-6]
    	\arrow["{\nu^k \times \id{}}", from=1-7, to=2-6]
    	\arrow["{\denote{B}\comp\pi_2}", from=1-7, to=2-8]
    \end{tikzcd}\]
    \caption{Polynomial induced by a term operator}
    \label{fig:term-polynomial}
\end{mdframed}
\end{figure*}

\begin{figure*}
\begin{mdframed}
    \[
    \begin{tikzcd}
    	{\prod_{1 \leq i \leq n} T_{\denote{A_i}(\pbf C)}(\Gamma \times \svc{\denote{A^i_j}(\pbf C)}_{1 \leq j \leq k_i})} & {T(\Gamma \times \svc{\denote{B_j}(\pbf C)}_{1 \leq j \leq k})} \\
    	{1} & {S}
    	\arrow["{\tau_{\Gamma  \times \svc{\denote{B_j}(\pbf C)}_{1 \leq j \leq k}
    	}}", from=1-2, to=2-2]
    	\arrow["{\denote{B}(\pbf C)}"', from=2-1, to=2-2]
    	\arrow[from=1-1, to=2-1]
    	\arrow["\denote{\oper{o}}^\sharp_\Gamma", from=1-1, to=1-2]
    \end{tikzcd}
    \qquad\qquad
    (\pbf C \in S^m)
    \]
    \[
    \text{Natural in the context $\Gamma$, where 
    $\svc{D_1, \ldots, D_\ell} 
    \defeq 
    (\cdots (\epsilon \times \svc{D_1}) \times \cdots) \times \svc{D_\ell}.$}
    \]
    \caption{Algebra structure induced by a term operator}
    \label{fig:term-algebra}
\end{mdframed}
\end{figure*}

\begin{definition}
A \emph{term operator signature}, denoted $\optm$, for an equational type signature $\sigty$ is given by a list of pairs of natural numbers $m \in \Nat$ and $\msigty(\underline{m})$-sorted second-order arities.
\end{definition}

\begin{example}[Term operators for the \stlc{}]
\setlength{\abovedisplayskip}{-1ex}
\begin{align*}
    &\rhd \oper{u} : \oper{Unit} \\
    A, B : * &\rhd \oper{abs} : \bind{A}{B} \to \oper{Fun}(A, B) \\
    A, B : * &\rhd \oper{app} : \oper{Fun}(A, B), A \to B \\
    A, B : * &\rhd \oper{pair} : A, B \to \oper{Prod}(A, B) \\
    A, B : * &\rhd \oper{proj}_1 : \oper{Prod}(A, B) \to A \\
    A, B : * &\rhd \oper{proj}_2 : \oper{Prod}(A, B) \to B
\end{align*}
\end{example}

A term operator signature induces a polynomial (\resp{} polynomial functor), given by taking the coproduct of the polynomials (\resp{} polynomial functors) induced by its elements.

\begin{notation}
We will denote by $\optm$ both a term operator signature and the polynomial functor it induces.
\end{notation}

\begin{remark}
\label{remark:explicit-dependent-products}
To gain intuition for the polynomial algebraic structure, it is instructive to evaluate the polynomials oneself, starting with an arbitrary $\tau : T \to S$ and taking pullbacks, dependent products and postcomposing. Of these operations, pullbacks and postcomposing are straightforward. 
We give the two relevant calculations for the dependent products explicitly.
\[\begin{array}{l}
\dprod{\codiag_n}(P \to \coprod_{1 \leq i \leq n} S^m)
\equiv 
  \big(\coprod_{\pbf A \in S^m} \prod_{1 \leq i \leq n} P_{\tuple{i, \pbf A}} \big) \to S^m
\\[10pt]
\dprod{\nu^k \times \id{}}(V^k \times P \xto{\id{}\times p} V^k \times S^m) 
\\[5pt]
\quad \equiv 
  \big(
    \coprod_{\pbf A \in S^k, \pbf B \in S^m} 
      {P_{\pbf B}}
        ^{ \prod_{1 \leq i \leq k} 
             V_{A_i}
          } 
  \big) 
  \to S^k \times S^m 
\end{array}\]
\end{remark}

\begin{proposition}
In elementary terms, $\optm$-algebras for the polynomial functor as in~\Cref{fig:term-polynomial} are equivalently given by typed term structures $\tau : T \to S$ with a natural transformation $\denop{o}^\sharp$ such that the diagram in \Cref{fig:term-algebra} commutes.
\end{proposition}

\begin{proposition}
\label{lemma:sigtm-finitary}
For all term signatures, the endofunctor $\optm$ on $\spsh C/S$ is finitary.
\end{proposition}

Thus, $\optm$ induces a monad~\cite{fiore2009construction} describing the term structure, closed under the operators of the signature.

\begin{notation}
Given a term operator signature $\optm$, we denote by $\moptm$ the free $\optm$-algebra monad on $\spsh C/S$.
\end{notation}

The Eilenberg--Moore category of the monad $\moptm$ is isomorphic to the category of $\optm$-algebras.

\section{Models of simply typed syntax}
\label{sec:models-syntax}

We now give the definition of simply typed syntax, along with its models. Note that this is not yet a full notion of simple type theory, as we lack substitution and equations.

\begin{definition}
A \emph{simply typed syntax} consists of:
\begin{itemize}
    \item a type operator signature $\opty$;
    \item a term operator signature $\optm$ for $\opty$.
\end{itemize}
\end{definition}

\begin{definition}
A \emph{model} for a simply typed syntax consists of
\begin{itemize}
    \item an $\opty$-algebra $\denop{ty} : \opty(S) \to S$;
    \item a cartesian $\opty$-typed context structure $\sct C$ for $S$;
    \item an $\optm$-algebra 
    $\denop{tm}\!: \optm {(\tau\!:\!T\!\to\!S)} \to {(\tau\!:\!T\!\to\!S)}$.
\end{itemize}
\end{definition}

\begin{proposition}
\label{prop:simply-typed-cwfs}
$S$-sorted simply-typed categories with families~\cite{castellan2019categories} are equivalent to models of simply typed syntax for an empty type and term signature, such that the carriers of the $\opty$- and $\optm$-algebras are $S$ and $\nu : V \to S$ respectively.
\end{proposition}

To discuss the relationships between different models of a simply typed syntax, and to prove that the syntactic model is initial, we need a notion of homomorphism. This necessarily involves a compatibility condition between algebraic term structures.

\begin{figure*}
\begin{mdframed}
\[\begin{tikzcd}[row sep=huge]
	{\textstyle
	   \prod_{1 \leq i \leq n} 
	     T_{\denote{A_i}(\pbf C)}
	      ( \Gamma \times \tuple{ \denote{A^i_j}(\pbf C) }_{1\leq j\leq k_i} )
	 } 
	& 
	{\displaystyle 
	   T_{\denote{B}(\pbf C)}
	     ( \Gamma \times \tuple{ \denote{B_j}(\pbf C) }_{1\leq j\leq k} )
	 } 
	\\
	{\textstyle
	   \prod_{1 \leq i \leq n} 
	     T'_{\denote{A_i}'(h(\pbf C))}
	       ( H(\Gamma) \times \tuple{ \denote{A^i_j}'(h(\pbf C)) }_{1\leq j\leq k_i} )
	 } 
	& 
	{\textstyle 
	   T'_{\denote{B}'
	     (h(\pbf C))}(H(\Gamma) \times \tuple{ \denote{B_j}'( h(\pbf C))}_{1\leq j\leq k} )
	  }
	\arrow["{\denop{tm}^\sharp_\Gamma}", from=1-1, to=1-2]
	\arrow["{(\denop{tm}')^\sharp_\Gamma}", from=2-1, to=2-2]
	\arrow["{
	  \prod_{1 \leq i \leq n} \displaystyle
	    (f_{\denote{A_i}(\pbf C)})_{(\Gamma \times \tuple{ \denote{A^i_j}(\pbf C) }_{1\leq j\leq k_i} )}
	  }" description, from=1-1, to=2-1]
	\arrow["{\displaystyle
	  (f_{\denote{B}(\pbf C)})_{(\Gamma \times \tuple{ \denote{B_j}(\pbf C) }_{1\leq i\leq k} )}
	  }" description, from=1-2, to=2-2]
\end{tikzcd}\]
\[
  \pbf C = \tuple{C_1,\ldots,C_m} \in S^m 
  \qquad\qquad 
  h(\pbf C) \defeq \tuple{h(C_1), \ldots, h(C_m)}
\]
\caption{Elementary term algebra coherence}
\label{fig:elementary-term-algebra-coherence}
\end{mdframed}
\end{figure*}

Categorically, this is made somewhat difficult to express by the fact that two typed term structures for the same signature may be algebras for polynomial endofunctors on different presheaf categories, depending on their cartesian \text{$\opty$-typed} context structures. To reconcile them, we will make use of the following lemma.

\begin{lemma}
\label{lemma:poly-commutes-with-hom}
Let $(\sct{C}, \tau : T \to S)$ and $(\sct{C}', \tau' : T' \to S')$ be models, and let 
$( H: \sct C \to \sct C' , h : S \to S' )$ be a cartesian $\opty$-typed context structure homomorphism between them (\Cref{def:context-structure-hom}).
Then there is a canonical natural transformation as follows.
%
\[\begin{tikzcd}
	{\spsh{C'}/S'} & {\spsh{C}/S} \\
	{\spsh{C'}/S'} & {\spsh{C}/S}
	\arrow["{\optm}", from=1-2, to=2-2] \arrow["{\optm[']}"', from=1-1, to=2-1] \arrow["{\pb h \comp \precomp H}", from=1-1, to=1-2] \arrow["{\pb h \comp \precomp H}"', from=2-1, to=2-2] \arrow[Rightarrow, from=1-2, to=2-1, shorten <=2mm, shorten >=2mm]
\end{tikzcd}\]
\end{lemma}

\begin{definition}
A \emph{homomorphism of models} for a simply typed syntax, from a model $(\sct{C}, {\tau: T \to S})$ to a model $(\sct{C}', {\tau' : T' \to S'})$, consists of
\begin{itemize}
    \item a cartesian $\opty$-typed context structure homomorphism 
    $( H: \sct C \to \sct C', h : S \to S' )$;
    \item a natural transformation $f: T \to T'H$,
\end{itemize}
such that the following diagrams, term-type coherence (left) and term algebra coherence~(right), commute:
%
\[\begin{tikzcd}
	{T} & {T'H} \\
	{S} & {S'}
	\arrow["{h}"', from=2-1, to=2-2] \arrow["{\tau}"', from=1-1, to=2-1] \arrow["{\tau'H}", from=1-2, to=2-2] \arrow["{f}", from=1-1, to=1-2]
\end{tikzcd}
\qquad
\begin{tikzcd}
	{\optm(T)} && {\optm(\pb{h}(T'H))} \\
	&& {\pb{h}(\optm['](T')H)} \\
	{T} && {\pb{h}(T'H)}
	\arrow["{\denote{\mathsf{tm}}}"', from=1-1, to=3-1] \arrow["{f_{\lrcorner}}"', from=3-1, to=3-3] \arrow["{\optm( f_{\lrcorner})}", from=1-1, to=1-3] \arrow["{\text{(\Cref{lemma:poly-commutes-with-hom})}}", from=1-3, to=2-3] \arrow["{\pb{h}(\denote{\mathsf{tm}}' H)}", from=2-3, to=3-3]
\end{tikzcd}\]
where $f_\lrcorner$ is the mediating morphism as in the following diagram.
\vspace{-1.6ex}
%
\[\begin{tikzcd}
	{T} & {\pb{h}(T'H)} & {T'H} \\
	& {S} & {S'}
	\arrow["{f_{\lrcorner}}" description, from=1-1, to=1-2]
	\arrow["{\tau' H}", from=1-3, to=2-3]
	\arrow["{\tau}"', from=1-1, to=2-2]
	\arrow["{h}"', from=2-2, to=2-3]
	\arrow[from=1-2, to=1-3]
	\arrow["{\pb{h}(\tau' H)}" description, from=1-2, to=2-2]
	\arrow["\lrcorner"{very near start, rotate=0}, from=1-2, to=2-3, phantom]
	\arrow["{f}", from=1-1, to=1-3, bend left]
\end{tikzcd}\]
\end{definition}
The term algebra coherence diagram expresses that $f_\lrcorner$ is an $\optm$-algebra homomorphism. This equivalently expresses that $f$ is a form of term algebra heteromorphism as in the following diagram.
%
\[\begin{tikzcd}[column sep=tiny]
	{\optm(T)} \arrow{r}[yshift=1ex]{\optm(f_{\lrcorner})} & {\optm(\pb{h}(T'H))} \arrow{r}[yshift=1ex]{\text{(\Cref{lemma:poly-commutes-with-hom})}} & {\pb{h}(\optm['](T')H)} \arrow{r}[yshift=1ex]{}
	& {\optm['](T')H} \\
	{T} &&& {T'H}
	\arrow["{\denote{\mathsf{tm}}}"', from=1-1, to=2-1]
	\arrow["{\denote{\mathsf{tm}}'H}", from=1-4, to=2-4]
	\arrow["{f}"', from=2-1, to=2-4]
\end{tikzcd}\]
In elementary terms, this corresponds to the coherence condition expressed in \Cref{fig:elementary-term-algebra-coherence}.
This resolves the compatibility difficulty described earlier.

\begin{example}
For any cartesian $\opty$-typed context structure homomorphism 
$(H, h)$,
there is a canonical model homomorphism 
$(H, h, v)$
for $v : V \to V'H$ given by the action of $H$:
\[
v_\Gamma\big(\, A \, , \, {\rho:\Gamma\to\svc A} \,\big) 
\defeq{}
\big(\, h(A) \, ,\, {H(\rho) : H(\Gamma) \to \svc{h(A)}} \,\big)
\]
\end{example}

Models of simply typed syntax and their homomorphisms, for a simply typed syntax $O$, form a category $\sct{S}_O$.

\begin{theorem}
\label{thm:initiality-syntax}
$\sct{S}_O$ has an initial object.
\end{theorem}
\begin{proof}
Let $\denop{ty}: \opty(S) \to S$ be the initial $\opty$-algebra and let $\sct C$ be the free cartesian $\opty$-typed context structure on $S$ as in \Cref{prop:initial-cartesian-typed-context-structure}. 
The slice category $\spsh{C}/S$ is cocomplete and the polynomial endofunctor $\optm$ is finitary (\Cref{lemma:sigtm-finitary}).  Thus, we have an initial 
$\optm$-algebra $\denop{tm} : \optm(\tau\!:\!T \to S) \longrightarrow (\tau\!:\!T \to S)$. 
Then $M \defeq (\sct C, \denop{ty}, \tau\!:\!T \to S, \denop{tm})$ is a model for the signature $O$.

Let $(\sct C', \denop{ty}', \tau'\!:\!T' \to S',  \denop{tm}')$ be a model of simply typed syntax for the signature $O$. There is a unique $\opty$-homomorphism $h: S \to S'$, by initiality of $S$, and ${H: \sct C \to \sct C'}$ is uniquely determined by the freeness of $\sct C$. Furthermore, there is a unique $\optm$-homomorphism ${f: T \to T'H}$ satisfying the coherence conditions by the initiality of $\tau : T \to S$. Finally, $(H, h, f)$ is a unique model homomorphism and $M$ is therefore initial.
\end{proof}

The initial object in $\sct{S}_O$ is the \emph{syntactic model}.
Indeed, according to the viewpoint of initial-algebra semantics~\cite{thomas1976initial}, syntactic models are precisely initial ones, for 
these have
a canonical compositional interpretation into all models and, as such, uniquely characterise any concrete syntactic construction up to isomorphism.  Here, it is further possible to make the finitary semantic construction of \Cref{thm:initiality-syntax} explicit to demonstrate its coincidence with familiar syntactic constructions.

\section{Substitution}
\label{sec:substitution}
\newcommand{\wk}{\oper{wk}}
\newcommand{\contr}{\oper{contr}}
\newcommand{\exch}{\oper{exch}}
\newcommand{\subs}{\oper{subst}}
\newcommand{\var}{\oper{var}}

$\optm$-algebras represent a notion of terms with (sorted and binding) algebraic structure. However, there are still two important concepts that are missing: that of (capture-avoiding) substitution, and that of equations. Substitution must be described before defining equations on terms, as many equational laws (such as the $\beta$-equality of the 
simply-typed \text{$\lambda$-calculus})
involve this meta-operation. Substitution is an important metatheoretic concept even besides this, and is necessary to define the multicategorical composition operation that will appear in some of the models of simple type theories (\Cref{def:multisubstitutional-model}).

To begin to talk about substitution, one must have a notion of \emph{variables as terms}, corresponding to the following structure.
\vspace{-1.6ex}
\[\begin{tikzcd}
	{V} & {T} \\
	& {S}
	\arrow["{\nu}"', from=1-1, to=2-2]
	\arrow["{\var}", from=1-1, to=1-2]
	\arrow["{\tau}", from=1-2, to=2-2]
\end{tikzcd}\]
This structure is not a term operator:
it may instead be added by considering free $\optm$-algebras on the typed term structure of variables, $\nu : V \to S$.

Substitution is traditionally given in one of two forms: single-variable substitution (typically denoted $\subst{t}{u}{x}$) and multivariable substitution (in which terms must be given for every variable in context). When the category of contexts is freely generated, these notions are equivalent.  In our more general setting single-variable substitution is the appropriate primitive notion.

One may present substitution as an operation given by the following rule.
\[\begin{prooftree}
\Hypo{\Gamma, x : A \vdash t : B}
\Hypo{\Gamma \vdash u : A}
\Infer2[$\oper{subst}$]{\Gamma \vdash \subst{t}{u}{x} : B}
\end{prooftree}\]
It corresponds to the polynomial below, according to the general description of \Cref{sec:algebraic-term-structure}.
\[S \xfrom{[\pi_2, \pi_1]} V \times S + S^2 \xto{[\nu \times \id{}, \id{}]} S^2 \xto{\pi_2} S\]
An algebra for the functor induced by this polynomial is given explicitly by a morphism $\subs$ in $\spsh C/S$ as in the diagram below.
%
\[\begin{tikzcd}
	{T} & {V \times T + T \times S} & {\coprod_{A, B \in S} {T_B}^{V_A} \times T_A} & {T} \\
	{S} & {V \times S + S^2} & {S^2} & {S}
	\arrow["{[\nu \times \id{}, \id{}]}"', from=2-2, to=2-3]
	\arrow["{[\pi_2, \pi_1]}", from=2-2, to=2-1]
	\arrow["{\pi_2}"', from=2-3, to=2-4]
	\arrow["{\tau}"', from=1-1, to=2-1]
	\arrow["{\tau}", from=1-4, to=2-4]
	\arrow[from=1-3, to=2-3]
	\arrow["{\subs}", from=1-3, to=1-4, dashed]
	\arrow["{[\pi_2, \pi_1]}"'{yshift=0.5ex}, from=1-2, to=1-1]
	\arrow["\lrcorner"{very near start, rotate=-90}, from=1-2, to=2-1, phantom]
	\arrow["{\tau \times \id{} + \tau \times \id{}}" description, from=1-2, to=2-2]
\end{tikzcd}\]
Here, for expository purposes, we shall equivalently consider the structure as given by a family of morphisms in $\spsh C$,
\[\subs_{A, B} : {T_B}^{V_A} \times T_A \to T_B \qquad (A,B \in S)\]
which is closer to the syntactic intuition.


The substitution operator must obey equational laws (\cf{} \cite[Definition~3.1]{fiore1999abstract} and \cite[Section~2.1]{fiore2014substitution}). This structure must be described semantically,
as it makes use of the implicit cartesian structure of the categories of contexts, which is not available syntactically.
Specifically, we require the following diagrams to commute. They correspond respectively to trivial substitution, left and right identities, and associativity.
%
\[\begin{tikzcd}
	{T_A \times T_B} & {{T_A}^{V_B} \times T_B} \\
	& {T_A}
	\arrow["{\wk \times \id{}}", from=1-1, to=1-2]
	\arrow["{\pi_1}"', from=1-1, to=2-2]
	\arrow["{\subs_{B, A}}", from=1-2, to=2-2]
\end{tikzcd}
\enspace
\begin{tikzcd}
	{1\times T_A} & {{T_A}^{V_A} \times T_A} \\
	& {T_A}
	\arrow["{\lambda(\var_A)\times\id{}}"{yshift=1ex}, from=1-1, to=1-2]
	\arrow["{\subs_{A, A}}", from=1-2, to=2-2]
	\arrow["\pi_2"', from=1-1, to=2-2, no head]
\end{tikzcd}\]
%
\[\begin{tikzcd}
	{{T_B}^{V_A} \times V_A} & {{T_B}^{V_A} \times T_A} \\
	& {T_B}
	\arrow["{\id{} \times \var_A}"{yshift=1ex}, from=1-1, to=1-2]
	\arrow["{\subs_{A, B}}", from=1-2, to=2-2]
	\arrow["{\contr}"', from=1-1, to=2-2]
\end{tikzcd}\]
%
\[\begin{tikzcd}
	{({T_C}^{V_B \times V_A} \times {T_B}^{V_A}) \times T_A} & {{T_C}^{V_A} \times T_A} 
	& 
	\\
	{({T_C}^{V_B \times V_A} \times T_A) \times ({T_B}^{V_A} \times T_A)} 
	& 
	{T_C}
	\\
	{({T_C}^{V_A \times V_B} \times {T_A}^{V_B}) \times ({T_B}^{V_A} \times T_A)} 
	& 
	{{T_C}^{V_B} \times T_B} &
	\arrow["{(\exch \times \wk) \times \id{}}", from=2-1, to=3-1]
	\arrow["{\subs_{B, C}^{V_A} \times \id{}}", from=1-1, to=1-2]
	\arrow["{\subs_{A, C}^{V_B} \times \subs_{A, B}}"', from=3-1, to=3-2]
	\arrow["{\subs_{A, C}}", from=1-2, to=2-2]
	\arrow["{\subs_{B, C}}"', from=3-2, to=2-2]
	\arrow["{\text{str}}", from=1-1, to=2-1]
\end{tikzcd}
\]
The morphism $\exch$ is given by $T^\gamma$ where $\gamma : X \times Y \xto{\equiv} Y \times X$ is the cartesian symmetry; $\wk$ by $X^! : X \to X^Y$; and $\contr$ by the evaluation.
By the extension structure of cartesian $\sigty$-typed context structures, they respectively correspond to the admissible syntactic operations of exchange, weakening, and contraction.  The map $\text{str}$ is the canonical strength of products and the map $\subs_{A,B}^P$ is the composite
\[
  {T_B}^{V_A\times P}\times {T_A}^{P}
  \xto{\equiv}
  ({T_B}^{V_A}\times T_A)^{P} 
  \xto{(\subs_{A,B})^{P}}
  {T_B}^{P}
\]

Crucially, substitution must also commute with all the operators of the theory: for every \emph{unparameterised} operator $\oper{o}$ 
as in \Cref{fig:term-operator}, we require the following diagram to commute for all $\pbf C \in S^m$ and
$D\in S$, where $E_i \defeq \prod_{1 \leq j \leq k_i} V_{\dn{A^i_j}(\pbf C)}$.
%
\[\begin{tikzcd}
	{\big(\prod_{1 \leq i \leq n} {T_{\dn{A_i}(\pbf C)}}^{E_i}\big)^{V_D} \times T_D}
	& 
	{{T_{\dn B(\pbf C)}}^{V_D} \times T_D}
	\\
	{\prod_{1 \leq i \leq n} 
	  {T_{\dn{A_i}(\pbf C)}}^{V_D\times E_i} \times {T_D}^{E_i}} 
	\\
	{\prod_{1 \leq i \leq n} {T_{\dn{A_i}(\pbf C)}}^{E_i}} 
	& 
	{T_{{\dn B}(\pbf C)}}
	\arrow["{(
	  \denop{o}^\sharp)
	  ^{V_D} \times \id{}}"{yshift=1ex}, from=1-1, to=1-2]
	\arrow["{\subs_{D, \dn B(\pbf C)}}" description, from=1-2, to=3-2]
	\arrow["{\prod_{1 \leq i \leq n} \subs^{E_i}_{D,\dn{A_i}(\pbf C)}}", from=2-1, to=3-1]
	\arrow["{\denop{o}^\sharp}"', from=3-1, to=3-2]
	\arrow["{\cong \comp (\cong \times \tuple{\wk_i}_{1 \leq i \leq n})}", 
	  from=1-1, to=2-1]
\end{tikzcd}\]

Models of simply typed syntax with unparameterised term operators may be extended to incorporate variable and substitution structure; together with homomorphisms that preserve this structure, they form a category.

\begin{proposition}
For a term signature $\optm$ with unparameterised operators, models of simply typed syntax with variable and substitution structure over a fixed cartesian typed context structure admit free constructions, and thereby generate a free monad
which we denote by $\moptmsubst$.
\end{proposition}

\begin{proposition}
$\moptmsubst$-algebras for the free cartesian \text{$([\,], [\,])$-typed} context structure on a single sort (\Cref{prop:initial-cartesian-typed-context-structure}) are, equivalently, $\optm$-substitution algebras~
\cite{fiore1999abstract}.
\end{proposition}

\begin{theorem}[Substitution lemma, \cf{}~\cite{fiore1999abstract,fiore2008second}]
\label{thm:substitution-lemma}
For a fixed $\sigty$-algebra and the free cartesian $\sigty$-typed context structure
thereon, provided that the signature $\optm$ contains only unparameterised operators, the free $\moptm$-algebra on $\nu : V \to S$ and the initial \text{$\moptmsubst$-algebra} are isomorphic.
\end{theorem}

From the syntactic viewpoint, this means that substitution is admissible: adding a substitution operator to a simply typed syntax leaves the associated terms unchanged, because a term involving substitution is always equal to one that does not involve substitution.

\begin{proposition}
\label{prop:admissible-parameterised-operators}
In the presence of weakening and exchange (present in the simple type theories we consider here) and substitution, parameterised term operators (\Cref{def:parameterised-term-operators}) are admissible.
\end{proposition}

\section{Equations on terms}
\label{sec:equations}

Equations on terms may now be treated, analogously to those on types, with the proviso that one must keep track of sorts and variable contexts. In particular, we are interested in terms parameterised by a number of type metavariables, and term metavariables in
extended contexts~\cite{hamana2004free, fiore2008second}.
\begin{figure*}
\begin{mdframed}
\[\begin{prooftree}
\Hypo{\trec{\twt{x^1_1}{A^1_1}, \ldots, \twt{x^1_{k_1}}{A^1_{k_1}}}{t_1}{A_1}}
\Hypo{\cdots}
\Hypo{\trec{\twt{x^n_1}{A^n_1}, \ldots, \twt{x^n_{k_n}}{A^n_{k_n}}}{t_n}{A_n}}
\Infer3{\trec{\twt{y_1}{B_1}, \ldots, \twt{y_k}{B_k}}{l \jeq r}{B}}
\end{prooftree}\]
\caption{Natural deduction rule for a term equation}
\label{fig:term-equation}
\end{mdframed}
\end{figure*}
\begin{notation}
For ${m\in\nat}$, let $K_m$ denote the free $\sigty$-algebra $\msigty{(\underline m)}$ on a set of type metavariables $\underline m$ and let $\sct K_m$ denote the category of contexts $\freecart{K_m}$
of the free cartesian \text{$\sigty$-typed} context structure on a set of type metavariables $\underline{m}$ (\Cref{prop:initial-cartesian-typed-context-structure}).
\end{notation}

\begin{definition}
An \emph{$\optm$-term equation} is given by a triple 
\begin{equation}
\label{eq:term-equation}
\big(\,
     m \in \nat 
     \, , \, 
     \big(\mathbf{A}_1,\ldots\mathbf{A}_n\to\mathbf{B}\big) \in \ar_2\big(\mopty(\underline{m})\big) 
     \, , \, 
     ( l , r )
\,\big)
\end{equation}
with $l, r : \prod_{1 \leq j \leq k} V_{B_j} \to \moptmsubst( p : P \to K_m )_{B}$ a parallel pair of morphisms in $\widehat{\sct K_m}$ for
\begin{align*}
    P & \defeq \coprod_{1 \leq i \leq n} \prod_{1 \leq j \leq k_i} V_{A^i_j} &
    p \big( \tuple{ i , (\rho_1,\ldots,\rho_{k_i}) } \big) & \defeq A_i
\end{align*}
where $\mathbf{A}_i = (A^i_1, \ldots, A^i_{k_i})A_i$ and $\mathbf{B}=(B_1,\ldots,B_k)B$.
\end{definition}

The parallel pair equivalently corresponds to a pair of terms in $\moptmsubst(P)_{B}(\svc{B_1,\ldots,B_k})$
which may be syntactically presented as in \Cref{fig:term-equation}.

\begin{definition}
An \emph{equational term signature}, typically denoted $\sigtm$, is given by a term operator signature $\optm$ and a list $\eqtm$ of $\optm$-term equations.
\end{definition}

Fix an $\optm$-term equation as in~\eqref{eq:term-equation} and consider an \text{$\optm$-algebra} in $\spsh{\sct C}/S$:
\begin{equation}
\label{eq:optm-algebra}
    \denop{tm} : \optm(\tau\!:\!T \to S) \longrightarrow (\tau\!:\!T \to S)
\end{equation}
Every $\pbf C \in S^m$ freely induces a homomorphism $(H, h) : (\sct K_m, K_m) \to (\sct C, S)$, and every morphism $\mathbf t$ in $\psh{\sct K_m}/{K_m}$ as below
\begin{equation}
\label{eq:term-assignment}
\begin{tikzcd}[row sep=scriptsize, column sep=scriptsize]
    P & & h^*(TH) 
    \\
    & K_m &
    \arrow["\displaystyle \mathbf t", from=1-1, to=1-3]
    \arrow["\displaystyle p"', from=1-1, to=2-2]
    \arrow["\displaystyle h^*(\tau H)", from=1-3, to=2-2]
\end{tikzcd}
\end{equation}
freely induces the following situation, analogously to \eqref{eq:type-equation-substitution}.
\[\begin{tikzcd}
	{\optm(\moptmsubst( P ))} 
	&& 
	{\optm(\pb{h}(TH))} 
	\\
	{\moptmsubst( P )} 
	&& 
	{\pb{h}(TH)} 
	\\
	& P 
	&&
	\arrow["{\eta}", from=3-2, to=2-1]
	\arrow["{\psi_{\mathbf t}}"', from=2-1, to=2-3, dashed]
	\arrow["{\denop{tm}^\oast}"', from=1-1, to=2-1]
	\arrow["{\pb{h}(\denote{\mathsf{tm}} H)\ \comp\ \text{(\Cref{lemma:poly-commutes-with-hom})}}", from=1-3, to=2-3]
	\arrow["{\optm(\psi_{\mathbf t})}", from=1-1, to=1-3]
	\arrow["\mathbf t"', from=3-2, to=2-3]
\end{tikzcd}\]

\begin{definition}
An $\optm$-algebra as in \eqref{eq:optm-algebra} satisfies an \text{$\optm$-term} equation as in~\eqref{eq:term-equation} whenever, for all $\pbf C$ and $\mathbf t$ as in the preceding discussion,
$\psi_{\mathbf t} \comp l = \psi_{\mathbf t} \comp r$.
\end{definition}

A morphism $\mathbf t$ as in~\eqref{eq:term-assignment} corresponds to a family
\[
t_i \in T_{\dn{A_i}(\pbf C)}(\tuple{\dn{A^i_1}(\pbf C),\ldots,\dn{A^i_{k_i}}(\pbf C)})
\qquad 
(1\leq i\leq n)
\]
As such, it provides a valuation for the term placeholders of the terms in the equation. Indeed, the evaluation of $\psi_{\mathbf t}$ at ${u \in \moptmsubst(P)_{B}(\tuple{B_1,\ldots,B_k}})$ is the 
term 
resulting from a meta-substitution operation replacing the term placeholders in $u$ with the concrete terms
$(t_i)_{1 \leq i \leq n}$.

\begin{definition}
Given an equational term signature $\sigtm = (\optm, \eqtm)$, a \emph{$\sigtm$-algebra} is an $\optm$-algebra that satisfies the equations of $\eqtm$.
\end{definition}

Equational term signatures (like equational type signatures) are an entirely syntactic notion and correspond exactly to systems of natural deduction rules presenting a simple type theory. We give examples.

\begin{example}
\label{eg:equational-presentations}
Equational presentations in multisorted universal algebra \cite{birkhoff1970heterogeneous} are examples of equational term signatures, whose operators are nonbinding and unparameterised.
\end{example}

\begin{notation}
We will informally denote by ${t : {(x\!:\!A)}B}$ a term metavariable $t$ of type $B$ in contexts extended by a fresh variable $x$ of type $A$, reminiscent of the notation for second-order arities (\Cref{not:second-order-arity}). The types of bound variables in term operators, and of terms themselves, may be inferred, and are elided.
\end{notation}

\begin{example}[$\beta$/$\eta$ rules for the \stlc{}]
\begin{align*}
    A, B : * \rhd t : (x\!:\!A)B, a : A 
    & \vdash 
    \oper{app}\big(\oper{abs}\big(\bind{z
      }{\subst{t}{z}{x}}\big), a\big) 
    \jeq 
    \subst{t}{a}{x} 
    \\[2.5pt]
    A, B : * \rhd f : \oper{Fun}(A, B) 
    & \vdash 
    \oper{abs}\big(\bind{x
      }{\oper{app}(f, x)}\big) 
    \jeq 
    f
\end{align*}
\end{example}

{
\allowdisplaybreaks
\begin{example}[Computational $\lambda$-calculus \cite{moggi1991notions}]
\label{example:computational-lambda-calculus}
The following extends the \stlc{}.
\begin{align*}
    & \rhd \oper{T} : * \to * 
    \\[.5em]
    A : * & \rhd \oper{return} : A \to \oper{T}(A) 
    \\[.25em]
    A, B : * & \rhd \oper{bind} : \oper{T}(A), \bind{A}{\oper{T}(B)} \to \oper{T}(B) 
    \\[.5em]
    A, B : * & \rhd a : A, f : \bind{x\!:\!A}{\oper{T}(B)} 
    \\ 
    & \!\!\!\!\!\!\!\!\!\!\!\!
    \vdash 
      \oper{bind}\big(\oper{return}(a), \bind{z
        }{\subst f z x}\big) 
      \jeq 
      \subst{f}{a}{x} 
    \\[.25em]
    A : * & \rhd m : \oper{T}(A) 
    \vdash 
      \oper{bind}\big(m, \bind{x
        }{\oper{return}(x)}\big) 
      \jeq 
      m 
    \\[.25em]
    A, B : * & \rhd m : \oper{T}(A), f : \bind{x\!:\!A}{\oper{T}(B)}, g : \bind{y\!:\!B}{\oper{T}(C)} 
    \\ & \!\!\!\!\!\!\!\!\!\!\!\!
    \vdash
      \oper{bind}\big(m, \bind{a
        }{\oper{bind}\big(\subst{f}{a}{x}, \bind{b}{\subst g b y}\big)}\big) 
    \\
    & \!\!\!\!\!
      \jeq 
      \oper{bind}\big(\oper{bind}(m,(a)\subst f a x),(b)\subst g b y\big)
\end{align*}
\end{example}
}

Models of simply typed syntax with variable and substitution structure may be restricted to algebras for equational term signatures.

\begin{proposition}
\label{prop:msigtmsubst}
Algebras for equational term signatures $\sigtm$ with unparameterised operators over a fixed cartesian typed context structure admit free constructions, and thereby generate a free monad, which 
we denote by $\msigtmsubst$.
\end{proposition}

The monad associated to an equational term signature $(\optm, [\,])$ is the free $\optm$-monad with variable and substitution structure $\moptmsubst$.  For any list of $\optm$-term equations $\eqtm$, there is a canonical quotient monad morphism $\moptmsubst \epito \msigtmsubst$.

\section{Models of simple type theories}
\label{sec:models-theories}

Simple type theories extend simply typed syntax by incorporating variable, substitution, and equational structure.

\begin{definition}
A \emph{simple type theory} consists of:
\begin{itemize}
    \item an equational type signature $\sigty$;
    \item an equational term signature $\sigtm$ for $\sigty$.
\end{itemize}
\end{definition}

\begin{definition}
A \emph{model} for a simple type theory consists of
\begin{itemize}
    \item a $\msigty$-algebra $\denop{ty} : \msigty(S) \to S$;
    \item a cartesian $\sigty$-typed context structure $\sct C$ for $S$;
    \item a $\msigtmsubst$-algebra 
      $\denop{tm} : \msigtmsubst (\tau\!:\!T \to S) \to (\tau\!:\!T \to S)$.
\end{itemize}
In particular, the type and term algebras both satisfy the specified equations.
\end{definition}

\begin{definition}
    A \emph{homomorphism of models} for a simple type theory is a homomorphism $(H, h, f)$ for the underlying simply typed syntax such that $f$ preserves the variable structure and is a heteromorphism for the substitution structure.
\end{definition}

Models of simple type theories and their homomorphisms, for a simple type theory $\Sigma$, form a category $\sct{S}_\Sigma$.

\begin{theorem}
\label{thm:initiality-theory}
$\sct{S}_\Sigma$ has an initial object.
\end{theorem}

The initial object is the \emph{syntactic model}.  It is given by a construction analogous to the one in \Cref{thm:initiality-syntax}, taking \Cref{prop:msigtmsubst} into account.

\section{Classifying multicategories}
\label{sec:classifying-multicategories}

The classes of models we have considered so far are very general. First, contexts must be closed under extension, but may not necessarily be lists of sorts. More importantly, substitution is not inherent in simply typed syntax, which allowed us to consider models with and without substitution: it is only by making this distinction that we are able to prove metatheoretic properties regarding substitution, such as in \Cref{thm:substitution-lemma}. However, one typically wishes to consider simple type theories that do have an associated notion of substitution, along with contexts that are lists. In this setting, we can reformulate the models to be more familiar to the models dealt with in categorical algebra (see \eg{} \textcite{crole1993categories}).

\begin{definition}
\label{def:multisubstitutional-model}
A model of simple type theory is \emph{multisubstitutional} if
the embedding of the set of sorts in the category of contexts presents the latter as the strict cartesian completion of the former.
\end{definition}

Multisubstitutional models have list-like contexts and admit a multivariable substitution operation~\cite{fiore1999abstract} in addition to, and induced by, the single-variable substitution operation of \Cref{sec:substitution}.  In fact, we have the following result.

\begin{theorem}
\label{thm:equivalence-multicategories}
Multisubstitutional models of simple type theories with empty type operator signatures are equivalent to cartesian multicategories with corresponding
structure.
\end{theorem}

We sketch the idea. There is an equivalence taking such a multisubstitutional model 
$(\sct C, \tau : T \to S)$ to a cartesian multicategory $\sct M$, with object set $S$; multihoms 
$\sct M(A_1, \ldots, A_n; B) = T_B(\svc{A_1,\ldots,A_n})$; identities 
arising from
the variable structure; composition given by the multivariable substitution operation (or, equivalently, by iterated single-variable substitution); and cartesian multicategory structure given by the functorial action of the presheaf $T$ along the exchange, weakening, and contraction structure of 
$\sct C$. Model homomorphisms define cartesian multifunctors.

The algebraic structure on $T$ induces
structure on $\sct M$, where each term operator induces a pair of functors (corresponding to the premisses and conclusion), natural transformations between which correspond to 
interpretations of
the operator.

There are a variety of notions equivalent to cartesian multi\-categories, such as
many-sorted abstract clones and multi\-sorted Lawvere theories, giving corresponding versions of \Cref{thm:equivalence-multicategories} for each notion. Relevant for future work on polynomial models of dependent type theories is the relationship with categories with 
families~\cite{dybjer1995internal}.  We note that the relationship with simply-typed categories with families in \Cref{prop:simply-typed-cwfs} extends to incorporate operators and equations. However, care must be taken: in the context of categories with families, type theoretic structure is typically expressed through generalised algebraic theories, which permit operators that are not natural in a categorical sense; while, conversely, such unnatural operators are forbidden in the current framework.


Theorems~\ref{thm:initiality-theory} and~\ref{thm:equivalence-multicategories} provide a general systematic construction of the \emph{classifying cartesian multicategory} of any simple type
theory.
In the context of universal algebra, we have the following.

\begin{corollary}
The initial model of the simple type theory for an equational presentation in universal algebra is, equivalently, its abstract clone.
\end{corollary}

Beyond universal algebra, we have a kind of ``generalised Lambek correspondence'' between models of simple type theories and structured cartesian multicategories.  When the simple type theory has finite products, the classifying cartesian multicategory is representable and hence equivalent to a cartesian category.  In particular, we recover the classical Lambek correspondence.

\begin{corollary}[Lambek correspondence]
\label{cor:lambek}
The initial model of the simple type theory for the \stlc{} with a set of base types $B$ is, equivalently, the free cartesian-closed category on $B$. 
\end{corollary}

\bibliography{references}

\end{document}